\documentclass[envcountsame]{llncs}
\usepackage{graphicx}
\usepackage[noend]{algorithmic}
\usepackage{amsmath,mathtools}
\usepackage{amssymb}
\usepackage{multirow}
\usepackage{color}
\usepackage{xspace}
\usepackage{slashbox}

\usepackage{array}
\newcolumntype{K}[1]{>{\centering\arraybackslash}p{#1}}

\usepackage{makecell}

{\bfseries}{\itshape}
{\bfseries}{\itshape}
\newtheorem{lem}[theorem]{Lemma}{\bfseries}{\itshape}
{\bfseries}{\itshape}
\newtheorem{defn}{Definition}{\bfseries}{\itshape}
{\bfseries}{\itshape}

\newcommand{\ad}{\ensuremath{adopt}\xspace}
\newcommand{\ma}{\ensuremath{match}\xspace}
\newcommand{\ov}{\ensuremath{override}\xspace}
\newcommand{\wa}{\ensuremath{wait}\xspace}
\newcommand{\sa}{\ensuremath{n_a}\xspace}
\newcommand{\sh}{\ensuremath{n_h}\xspace}
\newcommand{\C}{\mathcal{C}}
\newcommand{\E}{\mathbb{E}}
\newcommand{\T}{\mathcal{T}}

\newcommand{\myitem}[1]{\noindent{\bf #1}}
\DeclarePairedDelimiter\ceil{\lceil}{\rceil}
\newcommand{\size}[1]{\left\vert#1\right\vert}

\newcommand{\future}[1]{\ensuremath{future\left(#1\right)}}
\newcommand{\past}[1]{\ensuremath{past\left(#1\right)}}

\newcommand{\ratio}{\left(\frac{\alpha}{1-\alpha}\right)}

\author{Yonatan Sompolinsky\inst{1} \and Aviv Zohar\inst{1,2}}

\institute{School of Engineering and Computer Science,\newline The Hebrew University of Jerusalem, Israel \and Microsoft Research, Herzliya, Israel \\ \email{$\left\{\right.$yoni\_sompo, avivz$\left.\right\}$@cs.huji.ac.il}
}

\date{}

\title{Bitcoin's Security Model Revisited}

\begin{document}
\maketitle
\begin{abstract}
	We revisit the fundamental question of Bitcoin's security against double spending attacks. While previous work has bounded the probability that a transaction is reversed, we show that no such guarantee can be effectively given if the attacker can choose when to launch the attack. Other approaches that bound the cost of an attack have erred in considering only limited attack scenarios, and in fact it is easy to show that attacks may not cost the attacker at all. We therefore provide a different interpretation of the results presented in previous papers and correct them in several ways. We provide different notions of the security of transactions that provide guarantees to different classes of defenders: merchants who regularly receive payments, miners, and recipients of large one-time payments. We additionally consider an attack that can be launched against lightweight clients, and show that these are less secure than their full node counterparts and provide the right strategy for defenders in this case as well. Our results, overall, improve the understanding of Bitcoin's security guarantees and provide correct bounds for those wishing to safely accept transactions. 
\end{abstract}

\section{Introduction}
Users of the Bitcoin system~\cite{SATOSHI} rely on the irreversibility of monetary transfers when using the currency. In particular, merchants that accept bitcoins, must be assured that once a payment has been accepted, it will not be reversed or routed to a different destination, and that they can safely dispense products and services in exchange for the funds. 

Payments may be rerouted or canceled if, for example, an attacker tries to send two conflicting transaction requests to the system in an attempt to send the same funds to two different destinations. The system cannot allow money to be used twice and thus one of the two conflicting payments must be rejected eventually. It is important that the recipient of the canceled payment is not fooled into thinking he has received the payment in the interim. Such an attack is called a \emph{double spending} attack. Indeed, Bitcoin's most important innovation is its solution to this very problem. 

Bitcoin's core data structure -- The Blockchain -- contains a record of all transactions that have been accepted by the system. Each block is a batch of accepted transactions that contains additionally the cryptographic hash of its predecessor in the chain, as well as a cryptographic proof-of-work. Blocks are created by nodes that solve this proof-of-work and in return collect fees from transactions embedded in their block and from newly minted money as well. These nodes are often called \emph{miners}.

In case several chains form, due to the concurrent action of miners, Bitcoin nodes accept the longest chain as the record of transactions that have occurred,\footnote{In fact the chain representing the highest cumulative amount of computational power is chosen. This is usually the longest chain.} and ignore transactions not contained in this chain. This re-selection of the set of accepted transactions may cause some payments to be canceled, which may be abused by an attacker. 
To be secure against such double spending, merchants are advised to wait until their transaction is included in a block, and that several blocks are built on top of it. 
The more blocks built atop a given block, the less likely it is that a conflicting longer branch will form (even under deliberate attempts).
For a transaction embedded in a block, the block containing it, and each block that follows on the main chain, is counted as an additional \emph{confirmation}.

Satoshi in his original work~\cite{SATOSHI}, as well as additional works that follow~\cite{MENI,GHOST,garay2015bitcoin}, offer a guarantee of the security of transactions in the currency. Specifically, each provides a similar theorem of the following ``flavour'':

\begin{theorem}[informal] 
	As long as the attacker holds less than 50\% of the computational power, and all honest nodes can communicate quickly (compared to the expected time for block creation), the probability of a transaction being reversed decreases exponentially with the number of confirmations it has received. 
\end{theorem}

This work is motivated by the following argument against any such guarantee: If the attacker is allowed to choose the time it prefers to transmit the transaction, no probabilistic guarantee can be given that its attack will fail. Indeed, the attacker
may try to create blocks prior to the the transmission of transactions, in a preparatory stage that we term \emph{pre-mining}. The pre-mining stage may take a long while to succeed, but once it does, an attack can be carried out with success-probability 1.\footnote{The attacker may, alternatively, settle for fewer blocks during the pre-mining stage, and then carry out an attack with lower probability.} Figure~\ref{fig::premining_attack} illustrates such an attack.

It is important to note that the pre-mining stage need not be costly to an attacker. In fact, attackers that employ selfish mining strategies~\cite{ES,OPT_SELFISH_MINING,STUBORN_MINING} repeatedly create secret chains that are longer than those of the network and which can be additionally used to launch double spending attacks, and in fact gain as they do so (provided that the attacker is sufficiently well connected to the network~\cite{ES}, or if delays exist~\cite{OPT_SELFISH_MINING}). 

\myitem{Our contributions.}
We discuss two pre-mining attacks that can be used when the attacker can choose the timing of the transaction: One attack is a generalization of the Finney attack~\cite{Finney}, and the other is a generalization of the (somewhat lesser known) Vector76~\cite{Vector76} attack. The first can be used effectively against any node, whereas the second, against nodes that do not broadcast blocks such as lightweight clients that do not maintain a full copy of the blockchain.

We propose four different versions of security guarantees (for regular nodes), given that pre-mining may in general take place:
\begin{enumerate}
	\item Defending an independently generated transaction (one whose timing does not rely on the attacker) \label{g1}
	\item Defending the long-term fraction of lost transactions to the merchant
	\item Defending all transactions from ever being double spent\label{g3}
	\item Upper bounding the average profit of the attacker during a continuing attack
\end{enumerate}

We formalize these notions in Section~\ref{sec::model}. We then introduce three families of acceptance policies, $\sigma^{arb}$, $\sigma^{frac}$, and $\sigma^{total}$ that provide a defense of according to guarantees~\ref{g1}-\ref{g3} above.

With respect to the bound on the profit of attackers, we show that, indeed, attackers with enough mining power (still under 50\%) or superior networking capabilities can profitably launch double spending attacks, even when selfish mining schemes alone are not profitable to them.

The first guarantee above most closely matches the flavour of the theorem given by Satoshi~\cite{SATOSHI} and following works. We provide a corrected analysis that better accounts for pre-mining in this case, and maintains the general exponential decay. We highlight that this result is of slightly lesser use, as in most cases, attackers can easily control the timing of an attack: a buyer, for example, can choose the exact moment at which it enters a store to buy items, since merchants usually provide continuous service. 

As it is impossible to bound well below 1 the probability that a transaction timed by the attacker will succeed, or to bound the cost for an attacker, we suggest the second security guarantee as an effective upper-bound on the losses experienced by a merchant who regularly transacts with the currency. 
This guarantee corresponds to a ``safety-level'' strategy for a merchant who wishes to ensure that only a small fraction of his accepted payments are double spent. Here, we compute the optimal attack for every given policy of the merchant and thus compute its exact safety level. 

The third model best applies to large valued transactions whose introduction to the blockchain may have been selected at the convenience of the attacker. We show that waiting for a fixed number of confirmations does not provide adequate security in this case. To circumvent this, we provide a policy that requires a number of confirmations logarithmic in the length of the chain, and prove that, with high probability, no transaction will \emph{ever} be attacked when sticking to this policy. The downside of this policy is of-course the fact that the number of confirmations that it requires grows (albeit incredibly slowly), as time goes on.

We now elaborate more on pre-mining attacks and why they pose a risk for a merchant that receives a transaction that was broadcast at a timing selected by the attacker.

\subsection{Pre-mining and selective attack timings}
\begin{figure*}
	{\centering
		\includegraphics[scale=0.5]{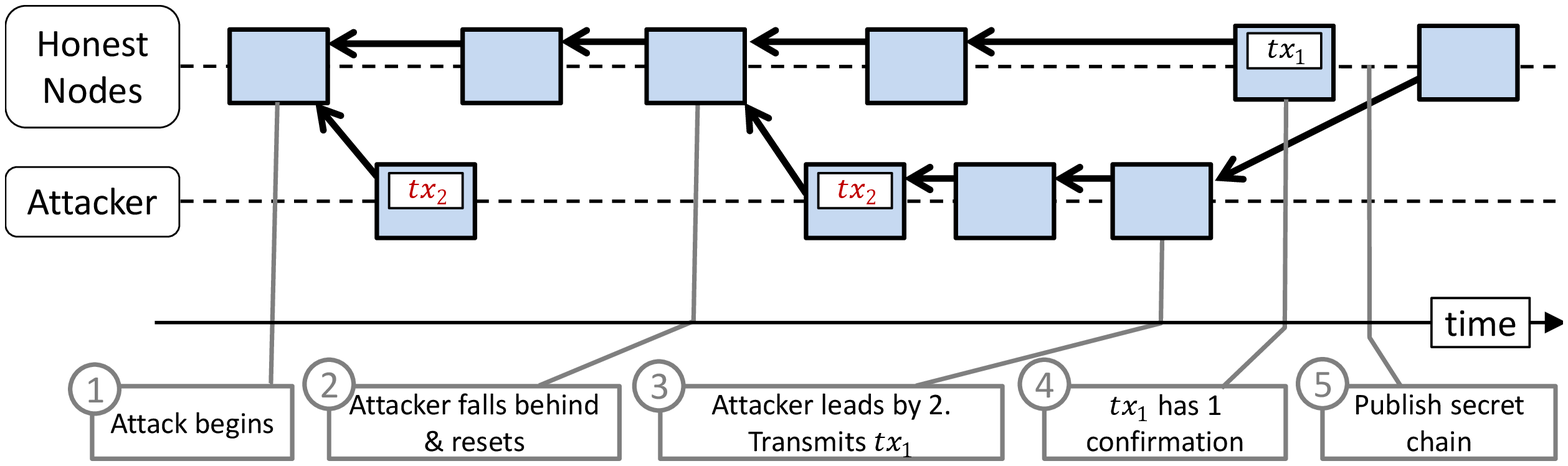}
	}	
		\emph{As the attack begins the attacker starts working on a secret chain with $tx_2$ inside its first block (1). If the attacker's chain is shorter than the honest nodes', the attacker gives up and restarts the attack (2). The attacker manages to gain a lead of 2 blocks (3). He then transmits the transaction he wishes to double spend which is included in a block (4).  The transaction now has enough confirmations (1-conf) and the attacker collects his rewards. He then publishes his secret chain and successfully double spends (5). Notice that once the pre-mining stage is concluded, the attack succeeds with probability 1, so miners that see $tx_1$ that is only broadcast then will always lose the funds.}
	\caption{\label{fig::premining_attack} The progression of a pre-mining attack on a 1-confirmation defender}
\end{figure*}

Consider the following attack scheme against a defender that waits for $k$ confirmations: 
In the pre-mining phase, the attacker begins to work on a secret branch that splits off from the most recent version of the chain. He embeds transaction $tx_2$ in this chain that conflicts the transaction $tx_1$ he wishes to double spend. If the attacker manages to create $k+1$ blocks more than the network, then he proceeds to carry out the attack. If at any point in time the network's chain is longer than the attacker's, he resets and starts a new branch, spiting off at a higher block in the public chain. Notice that this phase is in fact performed silently, and repeats until he is successful.
Once the attacker holds $k+1$ more blocks than the network's chain, he broadcasts his transaction to the network (when a large enough fee to ensure he is included in the next block), and waits for he to gain $k$ confirmations. It then releases his chain which is adopted immediately by the network, and invalidates the transaction $tx_1$. 

Observe that since the attack is only visible if the attacker is going to win, the recipient of funds can never be safe---\emph{conditioned on seeing the transaction}, the attack succeeds with probability 1. More sophisticated schemes are possible, and will be discussed throughout this paper.
The restricted version of this attack for a 0-confirmation defender is simply known as \emph{The Finney Attack} (named after its discoverer, Hal Finney, one of Bitcoin's first adopters). 
The key point in these pre-mining attacks is that they are not carried out at an ``arbitrary'' moment in time, but rather at a moment selected by the attacker. This leads us to the natural question: In what sense then is Bitcoin secure against such attacks? 

\subsection{Guarantees}
We consider three main scenarios that a recipient of funds may face:\\
\myitem{1. Protecting an independently placed transaction.} Here we assume a transaction has been placed in a block \emph{independently of the actions of the attacker}. One example of such a transaction is a minting transaction (also known as a \emph{coinbase transaction}) that the recipient may wish to accept. In this case, the attacker could not have chosen to launch the attack once he is successful in a pre-mining stage, but may still have pre-mined blocks.\footnote{This is the scenario that most closely matches previous results in~\cite{GHOST,SATOSHI,MENI,garay2015bitcoin}, although each work has analyzed it slightly differently. In this paper we augment the analysis with a proper quantification of the attacker's pre-mining.} \\
\myitem{2. Protecting a large fraction of blocks}
Here we consider a merchant that regularly receives payments in the blockchain and wishes to upper bound the loss he may suffer due to an attacker. We wish to find the attack policy that maximizes the fraction of blocks that are accepted by the network and then removed from the chain by the attacker. We note that a single double spending attack in which a long attack chain is released may remove many blocks simultaneously, thus this case differs from the previous one. An additional difference is that, in this case, attackers must actively decide when to give up on the attack on a specific block if the odds are not in his favour so that he may attack other more recent blocks instead. Such restarts are not considered when protecting a single transaction.

We show here, as well, that waiting for a fixed number of confirmations which is a function of $\epsilon$ can provide any level of security.

\myitem{3. Protecting all accepted blocks} In this case we consider a merchant that wishes to receive funds for a transaction at a moment in time that is possibly selected by the buyer. Given that the attacker can choose to place his transaction inside a block once he already knows he is certain to succeed, the only way to be fully secure in such cases is to find a policy for accepting transactions that never accepts a block that will be double spent.  As we have already discussed above, it is impossible to be secure against such a scenario by waiting for a fixed number of confirmations. While it is trivial to solve this problem by holding off acceptance of transactions indefinitely, we present a policy that guarantees that \emph{no block can be double spent} which fails only with arbitrarily low probability $\epsilon$, and requires a logarithmic number of confirmations in $1/\epsilon$ and the length of the chain. While waiting times that depend on the length of the chain and may grow are somehow unsatisfactory, we note that growth is extremely slow. Still, we believe that this is the main model that needs to be considered when dealing with extremely large transactions (e.g., when sums that are equivalent to tens of millions of USD are sent -- as was the case with several large bitcoin transactions like the FBI's seizure of the SilkRoad funds back in 2013).\\

\subsection{Related Work}
The first analysis of the resilience of Bitcoin is due to Nakamoto~\cite{SATOSHI}. His analysis considers a double spending attack without pre-mining, and is only approximate. Rosenfeld later goes on to correct the analysis~\cite{MENI}, and includes the pre-mining of a single block before it is launched (as such, it is not an attack against an arbitrarily chosen block). Rosenfeld further argues that the cost of an attack grows exponentially (which is correct as long as no block withholding is performed). Lewenberg et. al.~\cite{INCLUSIVE} also consider the exponential cost of the simple hidden-chain attack. 

In a previous work~\cite{GHOST} we have extended the analysis of security for settings with delay, demonstrating that the security of Bitcoin declines as delays increase and bounding its resulting throughput. We 
additionally present a time-dependent acceptance policy, applicable for the longest chain rule, that is more secure than purely structural policies (and as a result is faster to accept for a given level of security). In this work we restrict ourselves to structural policies, as it is not always guaranteed that the recipient remains online to time the creation of blocks. 

Garay et. al.~\cite{garay2015bitcoin} provide a formal model for the core of the Bitcoin protocol, using a discrete time setup where blocks can be found simultaneously at each step. They define desired properties of a blockchain protocol, and prove that they are satisfied by Bitcoin, when the attacker is adequately bounded.
They derive asymptotic bounds for security (and not an explicit formula). Their analysis assumes the transaction to defend is available to all honest nodes, and hence too roughly corresponds to an attack on an independently chosen block.

Karame et.al.~\cite{karame2012two} and Bambert et. al.~\cite{bamert2013have} have both considered double spending 0-confirmation payments. The simplest pre-mining attack which applies to such payments is known as the Finney attack~\cite{Finney}.

Eyal and Sirer~\cite{ES} suggested and analyzed a particular attack that knocks out blocks of the network in order to gain more from mining. Sapirshtein et. al.~\cite{OPT_SELFISH_MINING} improved the attack to optimal policies. This work uses these techniques to analyze optimal double spending strategies.
\section{The model}\label{sec::model}
We adopt the original setup analyzed by Satoshi Nakamoto, and later by Rosenfeld, that has become a standard model of Bitcoin's operation at the bound where block creation rates are much higher than the propagation time of blocks. 

Miners in the Bitcoin network create blocks with exponential inter-arrival times, with parameter $\lambda$ (in Bitcoin, lambda is 1/600 blocks/second).\footnote{$\lambda$ is in fact controlled via the difficulty of the proof-of-work that is embedded in each valid block, and the exponential inter-arrival time is a good approximation given that the proof-of-work is based on guessing inputs to a cryptographic hash function that will cause its output to land within some narrow range. This process is nearly memoryless.} Unless otherwise stated, we assume that honest nodes remain connected and can always communicate, and that the mining rate $\lambda$ remains constant over time.

Each block contains a reference to a single predecessor block (a cryptographic hash). The entire history of blocks created up to time $t$ forms thus a tree, which we denote by $\T^t$. Nevertheless, the Bitcoin protocol dictates that the valid history of transactions consists of (transactions in) the longest chain of blocks alone. Accordingly, we assume honest participants only keep track of the longest chain they have been presented with, and do not maintain the entire tree structure. We further assume that blocks propagate in the network very fast relative to $1/\lambda$, and under this assumption the honest network's chain at every point $t$ in time is uniquely determined; we denote it by $\C^t=\left(C^t_0,C^t_1,C^t_2,C^t_3,...,C^t_{height(t)}\right)$ ($height(t)$ is thus the length of the honest chain at time $t$). 
The block $C^t_0$ is a unique predetermined block that any chain must have at its root, and it is also called \emph{the genesis block}.
The height of a block $b$ is its distance from the genesis (with $height\left(C^t_0\right)=0$). We denote by $\past{b}$ ($\future{b}$) the set of blocks that precede (succeed) $b$ in the chain; note that $\future{b}$ keeps developing in time, as long as $b\in \C^t$.

The attacker is assumed to own an $\alpha$ fraction of the computational power, and the rest $(1-\alpha)$ is owned by honest nodes. Thus, the attacker creates blocks at a rate of $\alpha \cdot\lambda$, and the honest participants at a rate of $(1-\alpha)\cdot\lambda$. Following~\cite{OPT_SELFISH_MINING,ES}, we assume that the attacker has some communication capabilities that may allow it to transmit blocks that it has prepared in advance to nodes, just as the honest participants are starting to propagate a block that they have created. We denote the fraction of nodes that receives the attacker's block in this case by $\gamma\in [0,1]$. If $\gamma=0$, the attacker always loses block transmission races, and if $\gamma=1$ he wins them and is in fact able to get his block first to all honest miners.


Bitcoin nodes that participate in block creation efforts are called \emph{miners}. Upon mining a block, the miner embeds in it a set of Bitcoin transactions, created by users of the system. The transactions in $b$ must not double spend transactions in $\past{b}$.

\subsection{The acceptance policy}
A \emph{merchant} is any recipient, or beneficiary, of a Bitcoin transaction. 
Upon receiving a bitcoin transaction, the merchant considers it as either accepted (in case which it releases the good or service paid for), or not accepted---the latter, if it is not sufficiently convinced that it will remain forever in the longest chain. To decide this, the merchant, or ``defender'', uses an acceptance policy. We restrict our attention to acceptance policies that use only structural information that is available to a merchant that was offline, namely, the current chain $\C^t$.
We further assume that the defender is currently connected to the honest nodes, and is not isolated by an attacker. 

We define ``the acceptance policy'' of the defender as a function $\sigma_{\alpha,\gamma}:\mathbb{N} \to \mathbb{N}$. The function takes as input the height of the block containing the transaction, and returns the number of blocks that must be on top of it (including itself, i.e., $\size{\future{b}}+1$) before accepting the transaction. These blocks are often referred to as \emph{confirmations}. Thus, if the merchant sees a block $b$, and $\future{b}+1=\sigma_{\alpha}(h(b))$, then the transaction is considered accepted.

Perhaps the most commonly used policy for accepting transactions is the constant policy  $\sigma_{\alpha,\gamma}(h) \equiv k$ that requires a certain number of confirmations, independently of the location of the transaction in the chain. The number of confirmations, $k$, is often expressed as a function of $\alpha$ (and in our case $\gamma$ as well) to ensure the security of the chain. Our modeling, which allows for dependency on the block's height, will be justified in Section~\ref{sec::logarithmic}.

\subsection{The attack policy}
We follow~\cite{OPT_SELFISH_MINING} and define the attacker's policy as a function that determines the action of the attacker at every possible state. The attacker is assumed to be building a secret branch of the chain which he will use to later override the honest network's current chain. The attacker may take one of several actions:
\begin{itemize}
	\item \ad -- it abandons its attack, and its future chains will contain the tip of the honest network's current chain.
	\item \ov -- it overrides the honest network's current chain by publishing a strictly longer chain.
	\item \ma -- it publishes a chain of the same length as the honest network's current one.
	\item \wa -- the null action which waits for future events (i.e., block creations)
\end{itemize}

\subsection{Security properties}
We define three robustness notions that correspond to the different security guarantees suggetsed above.
For a block $b$, and acceptance policy $\sigma$, the event where $b$ is accepted by $\sigma$ is given by $\mathcal{E}_{accepted}(b) := \left\{\exists t: b=C^t_h \wedge height(t)\ge h+\sigma(h)\right\}$. Denote by $time_{accepted}(b)$ the time at which $\mathcal{E}_{accepted}(b)$ first occurred, with $time_{accepted}(b)=\infty$ if it never has occurred. The event where it is later removed from the longest chain is $\mathcal{E}_{attacked}(b):=\left\{\exists s: time_{accepted}(b)<s<\infty, b\notin \C^s\right\}$. 
\begin{defn}\label{def_arbitrary}
An acceptance policy $\sigma_{\alpha,\gamma}$ is $\epsilon$-arbitrary-robust, or robust against an attack on an arbitrary transaction, iff for any fixed height $h$ and $t$ such that $height(t)\ge h$, for any attacker with parameters $\alpha$ and $\gamma$, and under any attack policy $\pi_A$:\begin{equation}\label{eq_def_arbitrary}
\Pr\left(\mathcal{E}_{attacked}(C^t_h) \mid \mathcal{E}_{accepted}(C^t_h) \right)<\epsilon.\end{equation}
\end{defn}
\begin{defn}\label{def_fractional}
An acceptance policy $\sigma_{\alpha,\gamma}$ is $\epsilon$-fractional-robust, or robust against a removal of non-negligible portions of blocks, iff for any attacker with parameters $\alpha$ and $\gamma$, and under any attack policy $\pi_A$:\begin{equation}\label{eq_def_fractional}
\lim\limits_{t\rightarrow\infty}\frac{\sum_{b\in \T^{t}}\Pr\left(\mathcal{E}_{attacked}(b)\right)}{height(t)}<\epsilon.\end{equation}
\end{defn}
\begin{defn}\label{def_total}
	An acceptance policy $\sigma_{\alpha,\gamma}$ is considered $\epsilon$-totally-robust, or resilient to a double spend anywhere in the chain, iff for any attacker with parameters $\alpha$ and $\gamma$, under any attack policy $\pi_A$:
	\begin{equation}\label{eq_def_total}
	\Pr\left(\exists b\in \T^{\infty}:\mathcal{E}_{attacked}(b)\right)<\epsilon.\end{equation}
\end{defn}

In this paper, we introduce three families of acceptance policies, $\sigma^{arb}$, $\sigma^{frac}$, and $\sigma^{total}$ that are $\epsilon$-arbitrary-robust, $\epsilon$-fractional-robust, and $\epsilon$-totally-robust, respectively, for any $\epsilon>0$.

\section{Defending independently generated transactions} \label{sec::premining}
In this section we find policies that can be used by a defender to guarantee that transactions cannot be double spent, assuming that their timing cannot be controlled by the attacker. We show a strategy that waits for a constant number of confirmations (depending still on $\alpha$ and $\epsilon$) which is $\epsilon$-arbitrary-robust. In our analysis we fix some flaws in analysis done in previous works, specifically, we more precisely account for blocks mined before the attack.
Under the assumption that the attacker was \emph{not} involved in selecting the time at which the transaction appeared in the blockchain $\C^t$, we are able to provide a distribution over the number of blocks that it has prepared in advance and on any lead that it may have relative to the network (as it could not have conditioned the payment on some rare event). It is then possible to analyze when this transaction could be considered effectively irreversible.
This guarantee is applicable, for example, when the transaction is a minting transactions, whose timing is determined by the time at which the miner created a block. 

For the moment, we focus our analysis on the case $\gamma=0$. We define an acceptance policy $\sigma^{arb}_{\alpha}=\sigma^{arb}_{\alpha,0}$ as follows:
\begin{defn}
	Let $\sigma^{arb}_{\alpha}:=\min\left\{n\in\mathbb{N} : f(n,\alpha)<\epsilon\right\}$, \\where
\begin{align}\label{eq::number_of_conf}
f(n,\alpha):=&\sum\limits_{l=0}^{\infty}\frac{1-2\cdot\alpha}{1-\alpha}\cdot\left(\frac{\alpha}{1-\alpha}\right)^l\cdot \\ \nonumber &\left(\sum\limits_{m=0}^{n-l}\binom{m+n-1}{m}\cdot\alpha^m\cdot(1-\alpha)^n\cdot\left(\frac{\alpha}{1-\alpha}\right)^{n+1-m-l}+\right.\\&
\left.\sum\limits_{m=n-l+1}^{\infty}\binom{m+n-1}{m}\cdot\alpha^m\cdot(1-\alpha)^n\right)  \nonumber
\end{align}
\end{defn}

\begin{theorem}
For all $\epsilon>0$, the policy $\sigma^{arb}_{\alpha}$ is $\epsilon$-arbitrary-robust.
\end{theorem}
\begin{proof}
Let $b=C^t_h$ such that $len(t)\ge h$.
Let $C^{pub}_{t}$ be the longest chain in the published tree up to time $t$, and let $C^{oracle}_{t}$ be the longest chain in the entire tree, including secret attack-blocks. 
For any $z\in C^{pub}_{t}$, put $R^z_{t}:=\size{\future{z}\cap C^{oracle}_{t}\setminus C^{pub}_{t}}$ and $Q^z_{t}:=\size{\future{z}\cap C^{pub}_{t}}$. We call $\max_{z\in C^{pub}_{t'}}\left\{R^z_{t'}-Q^z_{t'}\right\}$ the \emph{pre-mined gap} of the attacker at time $t'$. In Lemma~\ref{lem::phase1} we show that a random variable $Y$ with distribution vector $\Pr(Y=n)=\frac{1-2\cdot\alpha}{1-\alpha}\cdot\left(\frac{\alpha}{1-\alpha}\right)^n$ stochastically dominates (first-order) the distribution of the random variable $\max_{z\in C^{pub}_{t'}}\left\{R^z_{t'}-Q^z_{t'}\right\}$.
Now, a necessary condition for a successful attack is that, above some $z\in\past{b}$, the attacker has managed to create a chain which is longer than the published chain above $z$ by the time when the attack is released (i.e., when the secret blocks are published). Therefore, the maximal pre-mined gap of the attacker over a block $z\in \past{b}$, by $time(b)$, can be upper bounded by a random variable with the distribution $(p_n)$. Since the creation of $b$'s predecessor, the attacker's gap over any $z\in\past{b}$ follows an ordinary random walk with drift towards negative infinity (with different $z$'s in $\past{b}$ corresponding to possibly different starting points, all bounded together by $(p_n)$). The probability that the attacker advanced $m$ blocks during the period at which the honest network created $n$ confirmation-blocks is given by $\binom{m+n-1}{m}\cdot(1-\alpha)^n\cdot\alpha^m$. The event in which the attack will succeed is then equivalent to the event that the walk will ever arrive at $X=-1$ (here we used the restriction to $\gamma=0$). For a given pre-mined gap of size $l$, this happens with probability 1, if $m>n-l$, and with probability $\left(\frac{\alpha}{1-\alpha}\right)^{m-n-l+1}$, if $m\le n-l$ (this can be derived, e.g., using a martingale method. See also in~\cite{MENI}). Altogether, we have thus shown that $f(n,\alpha)$, as defined in Equation~\ref{eq::number_of_conf}, upper bounds the probability that $b$ will ever be reversed.
\end{proof}

\begin{lem}\label{lem::phase1}
For a fixed time $t'$, 
$$\Pr\left(\max_{z\in C^{pub}_{t'}}\left\{R^z_{t'}-Q^z_{t'}\right\} \ge n \right)\le \left(\frac{\alpha}{1-\alpha}\right)^n.$$
\end{lem}
\begin{proof} 
If an attacker aims to maximize the value of $\max_{z\in C^{pub}_{t'}}\left\{R^z_{t'}-Q^z_{t'}\right\}$', then its optimal strategy is as follows: It begin mining at time $t=0$, right after the creation of the genesis block, and whenever $\sh<\sa$ it performs \ad, resetting the attack above the tip of $\sh$. To see that this is optimal, simply observe that if the honest network has a positive lead over the attacker at time $t$, then by adopting $z_{tip}$ the attacker will have at least as high as a gap (i.e., $R^z_{t'}-Q^z_{t'}$) over $z=z_{tip}$ than it would have had above any other $z\in\past{z_{tip}}$ had it not adopted $z_{tip}$.

Consider now a random walk on the non-negative integers with a reflecting barrier at the origin: At position $k$, the probability to move one step to the right is $\alpha$, and to the left is $(1-\alpha)$. At the origin, the probability to move to the right is $\alpha$ and to stay in place is $(1-\alpha)$. 
the transition probability matrix $P$ is accordingly. Denote by $Y_n:=\sum_{i=0}^{n}X_i$ the location of the walk after $n$ steps were made. The stationary distribution of the process $(Y_n)$ is $p_n:=\left(\frac{1-2\cdot\alpha}{1-\alpha}\right)\cdot \left(\frac{\alpha}{1-\alpha}\right)^{n}$, because if $Y$ is the distributed according to the limiting distribution then, for $n>0$: \begin{align*}
&\Pr\left(Y=n\right)=(1-\alpha)\cdot \Pr\left(Y=n+1\right)+\alpha\cdot\Pr\left(Y=n-1\right) = \\& (1-\alpha)\cdot p_{n+1}+\alpha\cdot p_{n-1} = (1-\alpha)\cdot \left(\frac{1-2\cdot\alpha}{1-\alpha}\right)\cdot \left(\frac{\alpha}{1-\alpha}\right)^{n+1} \\& +\alpha\cdot \left(\frac{1-2\cdot\alpha}{1-\alpha}\right)\cdot \left(\frac{\alpha}{1-\alpha}\right)^{n-1} = \ratio^n\cdot\frac{1-2\cdot\alpha}{1-\alpha}=p_n,
\end{align*}
and for $n=0$: $\Pr\left(Y=0\right)=(1-\alpha)\cdot \Pr\left(Y=0\right)+(1-\alpha)\cdot\Pr\left(Y=1\right)$, implying $\Pr(Y=0)=\frac{1-2\cdot \alpha}{1-\alpha}=p_0$.

Denote by $t_i$ the creation time of the $i$th block in $C^{oracle}_{t_{acc}}$. 
We claim that $\max_{z\in C^{pub}_{t_i}}\left\{R^z_{t_i}-Q^z_{t_i}\right\}$ has the same probability distribution as $Y_i$. We prove it by an induction on $i$. For $i=0$, $t_0=0$. At time $0$, following the creation of the $genesis$ block, the value of $\max_{z\in C^{pub}_{t_i}}\left\{R^z_{t_i}-Q^z_{t_i}\right\}=\left(R^{genesis}_0-Q^{genesis}_0\right)$ is 0, as $\size{\future{genesis}\cap C^{oracle}_0}=0$; and likewise $Y_0=0$. Assume we have proved this for $i$, and we now prove it for $i+1$. With probability $\alpha$, the attacker creates the block at time $t_{i+1}$. In that case, it increases by 1 $R^{z}_{t_i}-Q^{z}_{t_i}$ for every $z\in C^{pub}_{t_i}$, and in particular $\max_{z\in C^{pub}_{t_i}}\left\{R^z_{t_i}-Q^z_{t_i}\right\}$ increases by 1; likewise, $Y_{i}$ increases by 1 with probability $\alpha$, since this is the probability of the $(i+1)$th step being towards positive infinity.
With probability $(1-\alpha)$, the honest network created the $(i+1)$th block. In that case, $R^z_{t_i}-Q^z_{t_i}$ decreases by 1 for every block $z\in C^{pub}_{t_i}$, whereas for the new block, $R^{z_{t_{i+1}}}_{t_i}-Q^{z_{t_{i+1}}}_{t_i}=0-0=0$. Thus, the value of $\max_{z\in C^{pub}_{t_{i+1}}}\left\{R^z_{t_i}-Q^z_{t_i}\right\}$ is the maximum between $\max_{z\in C^{pub}_{t_{i}}}$ $\left\{R^z_{t_i}-Q^z_{t_i}\right\} -1$ and 0. Similarly, if $Y_i>0$ then with probability $(1-\alpha)$ the $(i+1)$th step is towards negative infinity decreasing its value by 1, whereas if $Y_i=0$ then with probability $(1-\alpha)$ the value of $Y_{i}$ remains 0 at te $(i+1)$th step.

As argued above, the attacker does not lose anything from beginning its attack above the genesis block (recall its aim is to maximize the success-probability and not to minimize costs). Moreover, the process $(Y_n)$ forms an ergodic process: It is aperiodic because at the origin there's a positive probability to stay in place, and it is positive recurrent since the walk has a positive biased towards negative infinity. 
Moreover, it can be shown that the stationary distribution of this process stochastically dominates (first-order) the distribution of $Y=(Y_n)$: $\Pr\left(Y=n\right) =\left(\frac{1-2\cdot\alpha}{1-\alpha}\right)\cdot\left(\frac{\alpha}{1-\alpha}\right)^n$.  In particular, $\Pr\left(\max_{z\in C^{pub}_{t'}}\left\{R^z_{t'}-Q^z_{t'}\right\} \ge n \right)\le \left(\frac{1-2\cdot\alpha}{1-\alpha}\right)\sum_{k=n}^{\infty}\left(\frac{\alpha}{1-\alpha}\right)^k=\left(\frac{\alpha}{1-\alpha}\right)^n$.
\end{proof}

Note that the analysis above assumed $\gamma=0$, as the success probbaility of an attack in case of a tie was given by $\frac{\alpha}{1-\alpha}$.
The bound can be generalized to be valid for all $\gamma>0$ by waiting for an additional confirmation. This completes the description of the family $\sigma^{arb}_{\alpha,\gamma}$.

The formulas provided in the definition of $\sigma^{arb}_{\alpha}$ do not give a good feel for the security for the results. Following~\cite{MENI,SATOSHI}, we present the results in a table for some representative values.
Table~\ref{table::premining} illustrates the number of confirmations needed for an attacker of various sizes. This is to be contrasted with the table appearing in~\cite{MENI}, where the author assumed a constant pre-mining of 1 block before the transaction is transmitted (thus the analysis there is not of an arbitrarily timed transaction).

\begin{table*}[ht]
\caption{The probability of a successful attack on an arbitrary block ($\epsilon-arbitrary-robustness$), given the attacker's hashrate ($\alpha$) and the number of confirmations the acceptance policy waits for ($conf$). The calculation includes consideration for pre-mining.}
\label{table::premining}
\centering
{\ \ }\\
{\footnotesize
\begin{tabular}{|K{10mm}|K{10mm}|K{10mm}|K{10mm}|K{10mm}|K{10mm}|K{10mm}|K{10mm}|K{10mm}|K{10mm}|K{10mm}|}
\hline  $\alpha\backslash conf$   & 1       & 2       & 3         & 4         & 5         & 6         & 7         & 8         & 9         & 10        \\ \hline 
2\%  & 0.24\%  & 0.02\%  & $\approx$0\% & $\approx$0\% & $\approx$0\% & $\approx$0\% & $\approx$0\% & $\approx$0\% & $\approx$0\% & $\approx$0\% \\ \hline 
6\%  & 2.16\%  & 0.42\%  & 0.09\%    & 0.02\%    & $\approx$0\% & $\approx$0\% & $\approx$0\% & $\approx$0\% & $\approx$0\% & $\approx$0\% \\ \hline 
10\% & 5.98\%  & 1.85\%  & 0.60\%    & 0.20\%    & 0.07\%    & 0.03\%    & $\approx$0\% & $\approx$0\% & $\approx$0\% & $\approx$0\% \\ \hline 
14\% & 11.66\% & 4.88\%  & 2.11\%    & 0.93\%    & 0.42\%    & 0.19\%    & 0.09\%    & 0.04\%    & 0.02\%    & $\approx$0\% \\ \hline 
18\% & 19.13\% & 9.94\%  & 5.32\%    & 2.90\%    & 1.60\%    & 0.89\%    & 0.50\%    & 0.28\%    & 0.16\%    & 0.09\%    \\ \hline 
22\% & 28.27\% & 17.33\% & 10.89\%   & 6.95\%    & 4.48\%    & 2.91\%    & 1.91\%    & 1.25\%    & 0.83\%    & 0.55\%    \\ \hline 
26\% & 38.90\% & 27.17\% & 19.36\%   & 13.97\%   & 10.17\%   & 7.45\%    & 5.49\%    & 4.06\%    & 3.01\%    & 2.23\%    \\ \hline 
30\% & 50.70\% & 39.33\% & 30.98\%   & 24.64\%   & 19.73\%   & 15.88\%   & 12.84\%   & 10.41\%   & 8.46\%    & 6.89\%    \\ \hline 
34\% & 63.23\% & 53.37\% & 45.55\%   & 39.14\%   & 33.81\%   & 29.31\%   & 25.49\%   & 22.21\%   & 19.39\%   & 16.95\%   \\ \hline 
38\% & 75.80\% & 68.45\% & 62.25\%   & 56.85\%   & 52.09\%   & 47.85\%   & 44.03\%   & 40.58\%   & 37.45\%   & 34.56\%   \\ \hline 
42\% & 87.35\% & 83.09\% & 79.31\%   & 75.86\%   & 72.68\%   & 69.72\%   & 66.95\%   & 64.33\%   & 61.83\%   & 59.44\%   \\ \hline 
46\% & 96.26\% & 94.88\% & 93.61\%   & 92.41\%   & 91.27\%   & 90.17\%   & 89.10\%   & 88.05\%   & 86.99\%   & 85.82\%   \\ \hline 
48\% & 98.98\% & 98.59\% & 98.23\%   & 97.88\%   & 97.54\%   & 97.21\%   & 96.88\%   & 96.54\%   & 96.15\%   & 95.60\%   \\ \hline 
50\% & 100\%   & 100\%   & 100\%     & 100\%     & 100\%     & 100\%     & 100\%     & 100\%     & 100\%     & 100\%     \\ \hline 
\end{tabular}
}
\end{table*}

\section{Defending the long-term fraction of double spent transactions}

In this section we focus on finding the level of fractional-robustness of policies $\sigma^{frac}_{\alpha,\gamma}$ that wait for some given constant number of confirmations. We investigate what fraction of all blocks that the policy accepts are later overridden.

The robustness of $\sigma^{frac}_{\alpha,\gamma}$ is not derived analytically, but is rather computed by an algorithm that finds the optimal attack policy. To compute the optimal attack, we follow the technique introduced in~\cite{OPT_SELFISH_MINING}, that encodes the decision problem of an attacker as a sequence Markov Decision Problems (MDPs). These then encode the action that the attacker takes at each state: for every length of attacker chain $\sa$ and length of honest chain $\sh$ (measured from the block they fork at), the attacker needs to decide whether it continues to build atop his chain, abandons his efforts and starts a new fork, or publishes blocks to succeed by one (or, with less success-certainty, match) the length of the network's chain thereby overriding it (provided he has enough blocks to do this). The transition and reward matrices are summarized in Appendix~\ref{appendix::model}.

The main difference from the algorithm presented in~\cite{OPT_SELFISH_MINING} is that the latter used this technique to compute optimal selfish mining attacks, and to maximize the number of attacker blocks in the chain. In contrast, we reward the attacker differently (as its objective here is different): The attacker is rewarded 1 unit for every successful block that the network accepted (i.e., had enough confirmations) and that the attacker managed to later remove from the chain. This reward is normalized by the number of all accepted blocks. Due to this normalization, the output of this computation equals the expected number of attacker blocks over the expected total number of accepted blocks, which in turn equals the left-hand side of~(\ref{eq_def_fractional}).

Recall that the parameter $\gamma$ encodes the probability that a chain is overridden when it is matched in length. Since honest nodes adopt the first chains that ehy receive, in case of ties, the ability of the attacker to push his block first to a significant fraction of the nodes dictates the chances that the next block will be built on top of its chain.

\begin{figure}[ht]
    \centering
    \includegraphics[scale=0.85]{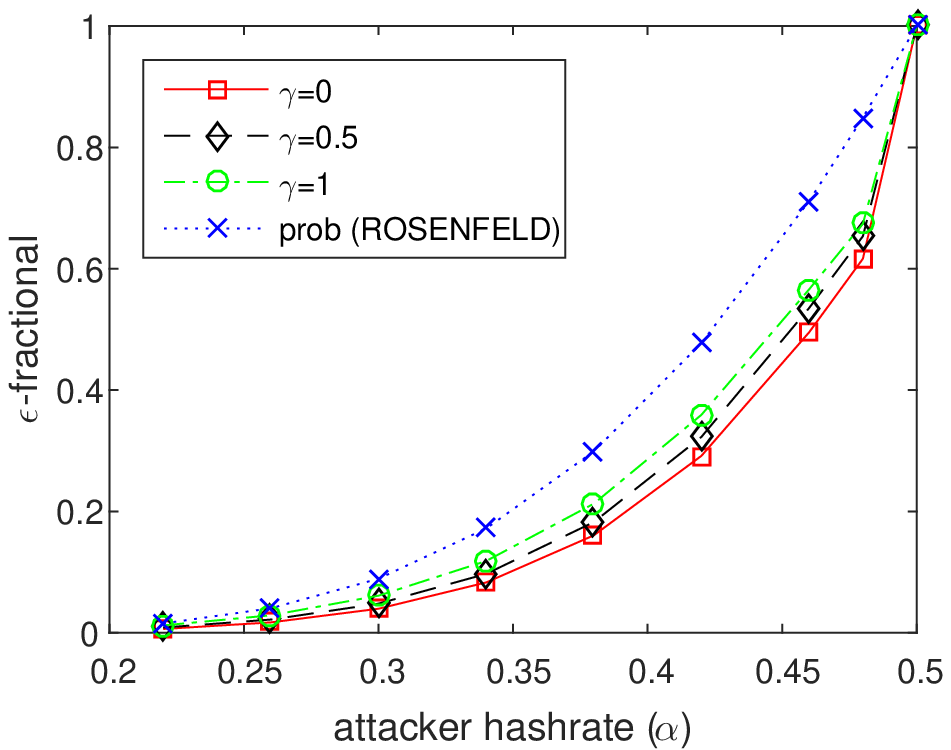}
  \caption{\label{fig::frac}The fraction of accepted blocks that an optimal attacker can double spend against a defender that uses 6 confirmations to accept as a function of the attacker's hashrate $\alpha$. The different curves correspond to different values of $\gamma$. Rosenfeld's result is also plotted for comparison.}
\end{figure}
Figure~\ref{fig::frac} depicts the results obtained for a policy with 6 confirmations, as computed on an MDP that was truncated to consider chains of length up to 60 blocks (the MDP analyzed in~\cite{OPT_SELFISH_MINING} is infinite and needs to be truncated for a numeric solution). The figure depicts the fraction of blocks an optimal attacker may double spend, for different values of $\gamma$. This essentially measures the $\epsilon$-robustness of the policy $\sigma\equiv6$.
The results of Rosenfeld for the probability of attack on a block~\cite{MENI} are included for comparison. It is interesting to note that the fraction of blocks that can be attacked is in fact lower than predicted by Rosenfeld (for any $\gamma$). This is because his analysis (and Satoshi's as well) consider an attack on a single block that goes on infinitely, that is, the attacker is assumed to never give up and to try to catch up with the chain no matter how far behind he is. In contrast, an attacker that aims to maximize the fraction of blocks it successfully attacks must occasionally give up and restart the attack if he is far behind. This effect is demonstrated in these results (note that, on the other hand, our model allows the attacker to double spend several blocks at once. These results demonstrate that the effective $\epsilon$ lowers nonetheless).

We similarly present the percentage of double spent blocks for different numbers of confirmations in Table~\ref{table::fractional}. Each cell was computed separately with its own optimal policy. 
\begin{table*}[ht]

\caption{The fraction of the network's blocks that an attacker with a given hashrate ($\alpha$) successfully attacks, when using an optimal attack policy, given the number of confirmations the acceptance policy waits for ($conf$).}
\label{table::fractional}
{\ \ }\\
\centering
{\footnotesize
\begin{tabular}{|K{10mm}|K{10mm}|K{10mm}|K{10mm}|K{10mm}|K{10mm}|K{10mm}|K{10mm}|K{10mm}|K{10mm}|K{10mm}|}
\hline
 $\alpha\backslash conf$     & 1       & 2           & 3           & 4           & 5           & 6           & 7           & 8           & 9           & 10          \\ \hline
2\%  & 0.08\%  & $\approx$ 0\% & $\approx$ 0\% & $\approx$ 0\% & $\approx$ 0\% & $\approx$ 0\% & $\approx$ 0\% & $\approx$ 0\% & $\approx$ 0\% & $\approx$ 0\%  \\ \hline
6\%  & 0.69\%  & 0.12\%      & 0.03\%      & $\approx$ 0\% & $\approx$ 0\% & $\approx$ 0\% & $\approx$ 0\% & $\approx$ 0\% & $\approx$ 0\% & $\approx$ 0\% \\ \hline
10\% & 1.89\%  & 0.52\%      & 0.16\%      & 0.05\%      & 0.02\%      & $\approx$ 0\% & $\approx$ 0\% & $\approx$ 0\% & $\approx$ 0\% & $\approx$ 0\% \\ \hline
14\% & 3.70\%  & 1.34\%      & 0.53\%      & 0.23\%      & 0.10\%      & 0.05\%      & 0.02\%      & $\approx$ 0\% & $\approx$ 0\% & $\approx$ 0\% \\ \hline
18\% & 6.16\%  & 2.75\%      & 1.34\%      & 0.69\%      & 0.36\%      & 0.20\%      & 0.11\%      & 0.06\%      & 0.04\%      & 0.02\%      \\ \hline
22\% & 9.37\%  & 4.92\%      & 2.80\%      & 1.66\%      & 1.02\%      & 0.64\%      & 0.41\%      & 0.27\%      & 0.18\%      & 0.12\%      \\ \hline
26\% & 13.47\% & 8.12\%      & 5.34\%      & 3.63\%      & 2.52\%      & 1.78\%      & 1.28\%      & 0.92\%      & 0.67\%      & 0.49\%      \\ \hline
30\% & 18.71\% & 13.20\%     & 9.63\%      & 7.19\%      & 5.48\%      & 4.23\%      & 3.30\%      & 2.60\%      & 2.07\%      & 1.66\%      \\ \hline
34\% & 26.46\% & 20.57\%     & 16.36\%     & 13.29\%     & 10.99\%     & 9.17\%      & 7.69\%      & 6.49\%      & 5.51\%      & 4.71\%      \\ \hline
38\% & 36.54\% & 31.04\%     & 26.95\%     & 23.60\%     & 20.77\%     & 18.37\%     & 16.39\%     & 14.66\%     & 13.16\%     & 11.84\%     \\ \hline
42\% & 50.32\% & 46.42\%     & 42.99\%     & 39.91\%     & 37.17\%     & 34.73\%     & 32.49\%     & 30.43\%     & 28.56\%     & 26.84\%     \\ \hline
46\% & 69.53\% & 67.65\%     & 65.84\%     & 64.06\%     & 62.33\%     & 60.64\%     & 59\%        & 57.38\%     & 55.79\%     & 54.23\%     \\ \hline
48\% & 81.48\% & 80.59\%     & 79.66\%     & 78.72\%     & 77.75\%     & 76.77\%     & 75.76\%     & 74.73\%     & 73.67\%     & 72.59\%     \\ \hline
50\% & 100\%   & 100\%       & 100\%       & 100\%       & 100\%       & 100\%       & 100\%       & 100\%       & 100\%       & 100\%       \\ \hline
\end{tabular}
}
\end{table*}

\subsection{Optimal policies}\label{subsec::optimal_policies} We now present the optimal policies returned by our algorithm, in two particular setups. Table~\ref{tbl::policy_26_0} describes the policy for an attacker with $\alpha=0.26, \gamma=0$ (here \ma is of no consequence). The row numbers correspond to the length of the attackers branch $\sa$ and the columns to the length of the honest network's branch $\sh$. 

Notice that here the attacker does not override the network's chain and receive rewards until its branch is of length three at least, as a successful attack requires the merchant sees three confirmations above its chain before the attack is released. Note, additionally, that the attacker does not give up on his attack when he is just slightly behind. If his chain is relatively long, he will not abandon it unless he is at least 3 blocks behind.

Table~\ref{tbl::policy_26_0.5} similarly corresponds to $\alpha=0.26, \gamma=0.5$. Each entry in it contains a string of three characters, corresponding to the possible status of the honest network: if it is working only on its own branch, if the attacker can possibly match the length of its branch and split its resources, and if it is already split between two branches of equal length ($fork$: $irrelevant,relevant,active$).\footnote{E.g., the string ``wm${*}$'' in entry $(\sa,\sh)=(3,3)$ reads: ``in case a fork is $irrelevant$ (that is, the previous state was $(2,3)$), $wait$; in case it is $relevant$ (the previous state was $(3,2)$), $match$; the case where a fork is already $active$ is not reachable''.} Actions are abbreviated to their initials: $\textbf{\emph{a}}dopt, \textbf{\emph{o}}verride, \textbf{\emph{m}}atch, \textbf{\emph{w}}ait$, while `$*$' represents an unreachable state.

\begin{table}[ht]
\caption{Optimal actions for an attacker with $\alpha=0.26,\gamma=0$, against a policy that waits for two confirmations, when the merchant accepts transactions after 2 confirmations. The row and column indices correspond to $\sa$ and $\sh$, respectively. 
	Actions are: $\textbf{\emph{a}}dopt, \textbf{\emph{o}}verride, \textbf{\emph{m}}atch, \textbf{\emph{w}}ait$, or `$*$' unreachable.}
 \label{tbl::policy_26_0}
{\ \ }\\
\centering
\begin{tabular}{|c|c|c|c|c|c|c|c|c|c|c|c|c|}
\hline
$\sa\backslash\sh$  & 0 & 1 & 2 & 3 & 4 & 5 & 6 & 7 & 8 & 9 & 10 \\ \hline
0& w & a & $*$ & $*$ & $*$ & $*$ & $*$ & $*$ & $*$ & $*$ & $*$\\ \hline
1& w & w & w & a & $*$ & $*$ & $*$ & $*$ & $*$ & $*$ & $*$\\ \hline
2& w & w & w & w & a & w & a & $*$ & $*$ & $*$ & $*$\\ \hline
3& w & w & o & w & w & w & w & a & $*$ & $*$ & $*$\\ \hline
4& w & w & w & o & w & w & w & w & a & $*$ & $*$\\ \hline
5& w & w & w & w & o & w & w & w & w & a & $*$\\ \hline
6& w & w & w & w & w & o & w & w & w & w & a\\ \hline
7& w & w & w & w & w & w & o & w & w & w & w\\ \hline
8& w & w & w & w & w & w & w & o & w & w & w\\ \hline
9& w & w & w & w & w & w & w & w & o & w & w\\ \hline
10 & w & w & w & w & w & w & w & w & w & o & w  \\ \hline
\end{tabular}
\end{table}


\begin{table}[ht]
\centering
\caption{Optimal actions for an attacker with $\alpha=0.26, \gamma=0.5$, for states $(\sa,\sh,\cdot)$ with $\sa,\sh\le 6$, when the merchant accepts transactions after 2 confirmations.
The row and column indices correspond to $\sa$ and $\sh$, respectively.
Actions are: $\textbf{\emph{a}}dopt, \textbf{\emph{o}}verride, \textbf{\emph{m}}atch, \textbf{\emph{w}}ait$, or `$*$' unreachable. The three entries at each cell are: $fork$: $irrelevant,relevant,active$}
\label{tbl::policy_26_0.5}

{\ \ }\\
\centering
\begin{tabular}{|c|c|c|c|c|c|c|c|}
\hline $\sa\backslash\sh$ & 0   & 1   & 2   & 3   & 4   & 5   & 6    \\ \hline
0  & w$**\!$& aa$*\!$& $*\!*\!*\!$& $*\!*\!*\!$&$*\!*\!*\!$&$*\!*\!*\!$&$*\!*\!*\!$\\  \hline
1  & w$**\!$&$*$w$*\!$& w$**\!$& a$**\!$&$*\!*\!*\!$&$*\!*\!*\!$&$*\!*\!*\!$\\ \hline
2  & w$**\!$& ww$*\!$& wm$*\!$& w$**\!$& a$**\!$&$*\!*\!*\!$&$*\!*\!*\!$\\ \hline
3  & w$**\!$& ww$*\!$& www & wm$*\!$& w$**\!$& a$**\!$&$*\!*\!*\!$ \\ \hline
4  & w$**\!$& ww$*\!$& www & wmw & wm$*\!$& w$**\!$& a$**\!$ \\ \hline
5  & w$**\!$& ww$*\!$& www & www & omw & wm$*\!$& w$**\!$ \\ \hline
6  & w$**\!$& ww$*\!$& www & www &$*$mw & omw & wm$*\!$ \\ \hline
\end{tabular}
\end{table}

\section{Logarithmic waiting time}\label{sec::logarithmic}
For every given acceptance policy one could ask what is the probability that at least one attack, in the course of the entire history, will be successful. This notion is formalized by $\epsilon$-total-robustness, in Definition~\ref{def_total}.
As discussed above, for any acceptance policy of the form $\sigma\equiv k$, for some constant $k$, the probability that a single attack on $\C^t$ will be successful goes to 1 as $t$ goes to infinity. Observe that achieving an arbitrary low $\epsilon$-total-robustness is trivially achievable by never accepting any transaction. Fortunately, below we show that there exists an $\epsilon$-totally-robust policy, for any $\epsilon>0$ which accepts every transaction in the blockchain after a time logarithmic in the chain's current length, as long as the block containing it still belongs to the longest chain. This result motivated the modeling of acceptance policies as taking the height of the block as an argument: It shows that considering policies not constant in the block height open up the option of achieving a strong security property unachievable otherwise.

\begin{theorem}
For any $\epsilon>0$, the policy $$\sigma^{total}_{\alpha,\gamma}\left(h\right):=C_{\alpha,\epsilon}+\lfloor\log_{b_{\alpha}}\left(h\right)\rfloor$$ is $\epsilon$-totally-robust, where  $C_{\alpha,\epsilon}:=\ceil*{\frac{1}{c}\cdot \ln\left(\frac1\epsilon\cdot b_{\alpha}\cdot\left(1-e^{-c}\right)^{-1}\cdot\left(1-b_{\alpha}/(e^c)\right)^{-1}\right)}$, $b_{\alpha}:=\frac{e^c+1}{2}$, with $c:=\frac18\cdot\frac{(1-2\cdot\alpha)^2}{1-\alpha}$.
\end{theorem}
\begin{proof}
\emph{Part I:}
Let us write $(\sa,\sh,h)$ whenever the attacker's chain is $\sa$ blocks long, the honest network's chain is \sh blocks long, and the earliest block in the chain of the network (i.e., the $\sh$-th block from the tip of the honest chain) is of height $h$.
By definition, the attacker can perform a successful attack iff $\sa\ge \sh$ and $\sh\ge \sigma^{total}_{\alpha,\gamma}\left(h\right)$. Assume now that if the attacker performs \ad at some state $(\sa,\sh,h)$ then the process transits to state $(0,0,h+1)$ (instead of $(0,0,h+\sh)$).\footnote{Technically, transiting to $(0,0,h+1)$ is realized by transiting to $(1,0,h+1)$ with probability $\alpha$ or to $(0,1,h+1)$ with probability $(1-\alpha)$.\label{footnote::technical}}
That this assumption works in favour of the attacker is clear: When $h'<h$, state $(\sh,\sa,h')$ is always preferable by the attacker over state $(\sh,\sa,h)$, for any $\sa,\sh$. Indeed, these two states differ only in that a successful attack in the latter state implies a successful attack in the former as well (but not necessarily vice versa), since the condition $\sh\ge \sigma^{total}_{\alpha,\gamma}\left(h\right)$ is stronger than $\sh\ge \sigma^{total}_{\alpha,\gamma}\left(h'\right)$. In particular, $(0,0,h+1)$ is preferable to the attacker over $(0,0,h+\sh)$.

Define $h_i=b^i\;(i=0,1,2....)$. Below we abbreviate $\sigma=\sigma^{total}_{\alpha,\gamma}$.
We define the $i$th epoch by the set of states with $h_i\le h<h_{i+1}$. The sequence $(h_i)$ satisfies the property that $\sigma\left(h_{i+1}\right)-1=\sigma\left(h_{i+1}-1\right)=\dots=\sigma\left(h_{i}\right)=C_{\alpha,\epsilon}+i$.
Let $p_i$ denote the probability that the attacker manages to perform a successful attack on a block belonging to the $i$th epoch. Suffice it to show that $\sum_{i=1}^\infty p_i<\epsilon$.

By definition, every attack in the epoch between $h_i$ and $h_{i+1}$ begins at a state of the form $(0,0,h)$ with $h_i\le h<h_{i+1}$ (an attack begins after the attacker abandoned its previous attempt, by an \ad, which leads to this state). A successful attack on block $u=C^t_h$ in the $i$th epoch can be reached only after at least $\sigma\left(h_i\right)$ blocks were built by the honest network above $b$ (as it requires $\sh\ge\sigma\left(h_i\right)$), including $u$, and in particular, after at least $\sigma\left(h_i\right)$ blocks were built since the creation of $u$'s predecessor (counting the blocks of both parties, and excluding $u$'s predecessor). For any number of steps $k\ge\sigma\left(h_i\right)$, the probability that after precisely $k$ steps $\sa\ge\sh$ is at most $e^{-c\cdot k}$: By putting $Z_i:=(X_i-1)/-2$, we arrive at a sequence of i.i.d random variables that take values in $\left\{0,1\right\}$. The event $\sum_{i=1}^{k}X_i\ge 0$ is then equivalent to $\sum_{i=1}^{k}Z_i\le k/2$, with $\mathbb{E}\left[\sum_{i=1}^{k}Z_i\right]=(1-\alpha)\cdot k$. We then apply to the latter sum Chernoff's bound: $\Pr\left(Z\le(1-\delta)\cdot \E\left[Z\right]\right)$, with $Z=\sum_{i=1}^{k}Z_i$ and $\delta:=\frac12\cdot\frac{1-2\alpha}{1-\alpha}$.

However, since the number of steps $k$ is not known in general, we upper bound the success probability of an attack on $u$ by $\sum\limits_{k=\sigma\left(h_i\right)}^{\infty}e^{-c\cdot k}=e^{-c\cdot \sigma\left(h_i\right)}\cdot\left(1-e^{-c}\right)^{-1}$.
By the union bound, the probability $p_i$ of a successful attack during the $i$th epoch can be upper bounded by $(h_{i+1}-h_i)\cdot e^{-c\cdot \sigma\left(h_i\right)}\cdot\left(1-e^{-c}\right)^{-1} < h_{i+1}\cdot e^{-c\cdot \sigma\left(h_i\right)}\cdot \left(1-e^{-c}\right)^{-1}=b^{i+1}\cdot e^{-c\cdot \sigma\left(h_i\right)}\cdot \left(1-e^{-c}\right)^{-1}$.

\emph{Part II:}
In order to upper bound the probability that there exists an epoch with a successful attack we apply the union bound on the entire sequence of epochs: 
\begin{align}
& \sum\limits_{i=0}^{\infty}p_i\le \left(1-e^{-c}\right)^{-1}\cdot\sum\limits_{i=1}^{\infty} b_{\alpha}^{i+1}\cdot e^{-c\cdot \sigma\left(l_i\right)} = \\&
b_{\alpha}\cdot\left(1-e^{-c}\right)^{-1}\cdot \sum\limits_{i=1}^{\infty} b_{\alpha}^i\cdot e^{-c\cdot (C_{\alpha,\epsilon}+i)}=
\\& b_{\alpha}\cdot\left(1-e^{-c}\right)^{-1}\cdot e^{-c\cdot C_{\alpha,\epsilon}}\cdot \sum\limits_{i=1}^{\infty} \left(b_{\alpha}/\left(e^c\right)\right)^i= \\ &
\label{eq::geometricsum}b_{\alpha}\cdot\left(1-e^{-c}\right)^{-1}\cdot e^{-c\cdot C_{\alpha,\epsilon}}\cdot \sum\limits_{i=1}^{\infty} (b_{\alpha}/e^c)^i <  \\ & \left(b_{\alpha}\cdot\left(1-e^{-c}\right)^{-1}\cdot\left(1-b_{\alpha}/(e^c)\right)^{-1}\right) 
\cdot e^{-c\cdot C_{\alpha,\epsilon}}<\epsilon.
\end{align}
The last inequality holds by the choice of $C_{\alpha,\epsilon}$, whereas the geometric series in~(\ref{eq::geometricsum}) converges due to $b_{\alpha}<e^c$, by the choice of $b_{\alpha}$.
\end{proof}

Observe that our analysis allowed the attacker to \ov whenever the random walk visits the zero (as we have bounded the probability that $\sum_{i}X_i\ge 0$). Consequently, this bound applies to an attacker of any $\gamma$-value, as it already assumes that the attacker is always able to \ma successfully.

\section{The Generalized Vector76 pre-mining attack}\label{sec::vec76}
In this section we present the generalized Vector76 attack (the original attack was suggested by a user named Vector76 in the bitcoinTalk forums to possibly explain a successful double spending attack against the MyBitcoin e-wallet~\cite{Vector76}).

The attack is a form of pre-mining attack that the attacker can work on in secret until he is guaranteed to be successful. In this case as well, hybrid methods that trade off a shorter preparation time in exchange for lower success probabilities exist. It further assumes that the victim is unable to relay blocks to the main network, e.g., if he uses a light weight client that receives cryptographic proof of the attack, but not the blocks themselves. The attack is then \emph{easier} to execute compared to a regular pre-mining attack, since it requires the attacker to generate less of a lead on the honest network, and in fact, in a reverse twist, relies upon the network to confirm the double spending payment.

The important aspect of this attack is that it draws a clear distinction between full nodes and light node implementations that do not relay blocks, and hence demonstrates that light nodes need in fact to wait for additional confirmations to be equally secure to full nodes.

The attack proceeds as follows:
\begin{enumerate}
	\item The attacker starts working on a secret branch of the chain. It embeds the transaction $tx_1$ (that it later wishes to reverse) in this block. 
	\item If the defender requires $\sigma\equiv k$ confirmations, the attacker needs to build an additional $k-1$ blocks on top of the one containing $tx_1$ (for a total of $k$ confirmations). He attempts to do so, in secret.
	\item If his branch of the chain is longer than that of the honest chain, at some point \emph{after} he has $k$ confirmations for $tx_1$, then it shows the $k$ confirmations to the lightweight client which then that accepts it as the legitimate chain, since it is the longest one.
	\item The attacker then transmits a conflicting transaction $tx_2$ to the honest network. As the honest network is not aware of the attacker's chain, the former's chain will grow long enough for $tx_2$ to be accepted by all nodes (and eventually even the attacked one). 
\end{enumerate}

Figure~\ref{fig::vector76_attack} depicts the attack. Again, notice that a crucial stage in the success of the attack is that the honest network does not adopt the block containing $tx_1$.

\begin{figure*}
	{\centering
		\includegraphics[scale=0.5]{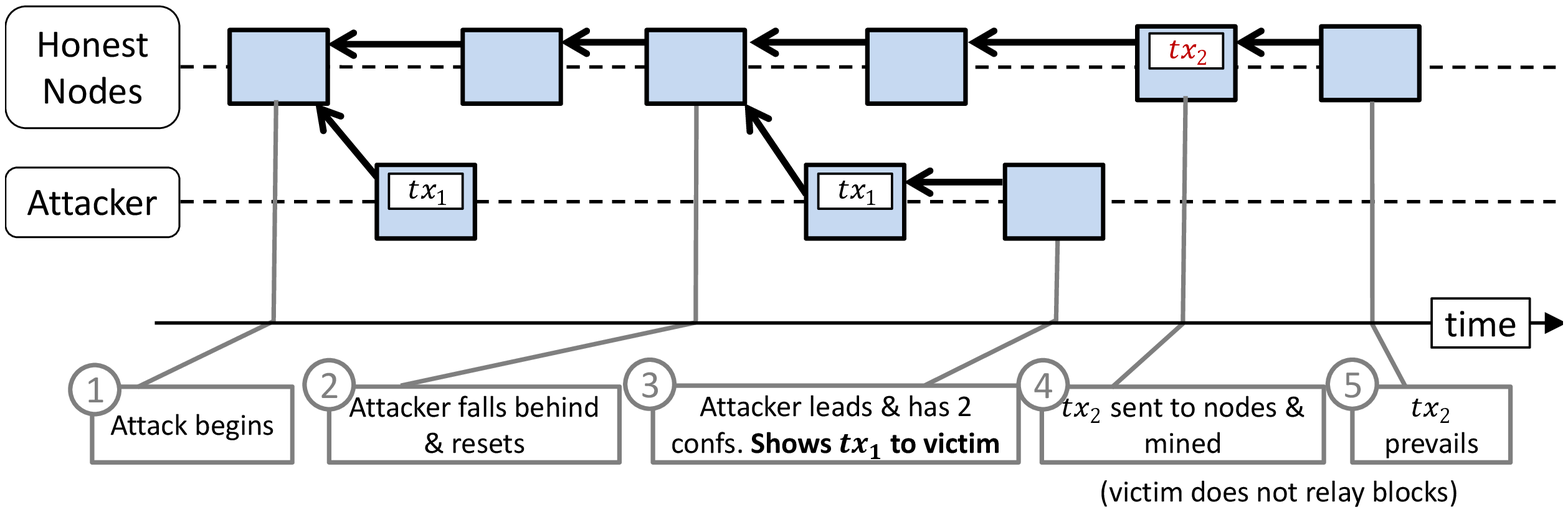}
	}	
	
	\emph{As the attack begins the attacker starts working on a secret chain with $tx_1$ inside its first block (1). If the attacker's chain is too far behind it may restart the attack (2). The attacker manages to gain a lead of 1 blocks, but has the two confirmations on his $tx_1$ needed to convince the victim (3). He then reveals the secret chain to the victim (that does not relay it), and collects an item in exchange. He then transmits the double spending transactions $tx_2$ to the network which is then included in a block (4).  The network continues to mine atop $tx_2$ and it eventually prevails (5).}
	\caption{\label{fig::vector76_attack} The progression of a generalized Vector76 attack on a 2-confirmation defender}
\end{figure*}

\myitem{Difficulty of the attack} The requirement for success in the pre-mining phase (after which the attack succeeds with probability 1) is that after constructing $k$ blocks or more, the attacker has a lead of 1 block over the network's chain. In a successful regular attack (whether it contains pre-mining or not), against a similar $k$-confirmation defender, the network has constructed $k$ blocks on top of the transaction (including the block that includes it), and at some point the attacker succeeds to lead over the network, and he thus has at least $k$ such blocks of his own in his branch. It is therefore easy to see that events in which successful double spending attacks occur are strictly contained in events in which the  generalized Vector76 attack is successful (and succeeds with probability 1). The above argument thus shows that light nodes are strictly less secure than regular nodes, and need to wait for more confirmations (we include below an analysis that quantifies the effect of waiting longer). 

In contrast, a regular (generalized Finney) pre-mining phase that
leads to a successful attack with probability 1 has much stricter requirements: the attacker needs to lead by $k+1$ blocks over the network (again, here he can launch a Vector76 attack as well, since he has built at least $k$ blocks and leads over the network).

\myitem{Resets and attacker strategy}
The generalized Vector76 attack described above can in fact be improved if attackers pick a better policy regarding restarts of  the attack. This policy can again be found by solving MDPs that reward successful attacks. We leave this for future work.

\myitem{The Vector76 attack against full nodes} The Vector 76 attack can also be applied to full nodes if the attacker can somehow manage to send his chain with $k$ confirmations to the victim while a similar length chain is being propagated through the network. In this case, even once the defender relays the block, the network will not adopt it, since it is of equal height to the one created by the attacker. The attacker naturally needs to time the transmission right, and if he misses, his original block may be adopted by some fraction of the honest nodes (a model with a parameter similar to $\gamma$ that is used in selfish mining may capture this). 

\myitem{Analytical guarantees of security}
Below we provide an analysis of the security of lightweight nodes against the generalize Vector76 attack. As the transactions that the node accepts are conditioned 
on the attack's current state, the arbitrary-block security guarantee does not apply. Fortunately, we can upper bound the safety level against any attack, as follows. 
The policy defined below is of the form $\sigma^{spv}_{\alpha}=\sigma^{spv}_{\alpha,\gamma}$, and applies for all $\gamma$.

\begin{defn}\label{def::vec76defense}
	Let $$\sigma^{spv}_{\alpha}:=\min\left\{k\in\mathbb{N} : g(k,\alpha)<\epsilon\cdot(1-\alpha)\right\},$$ where
\begin{align}\label{eq::vec76defense}
g(k,\alpha):=&\sum\limits_{l=0}^{\infty}\frac{1-2\cdot\alpha}{1-\alpha}\cdot\left(\frac{\alpha}{1-\alpha}\right)^l\cdot \\ \nonumber &\left(\sum\limits_{n=0}^{k+l}\binom{n+k-1}{n}\cdot\alpha^k\cdot(1-\alpha)^n+\right. \\ &\left.\sum\limits_{n=k+l+1}^{\infty}\binom{n+k-1}{n}\cdot\alpha^{n-l}\cdot(1-\alpha)^{k+l}\right) \label{eq::defened76}.
\end{align}
\end{defn}

\begin{theorem}
For any $\epsilon>0$, the policy $\sigma^{spv}_{\alpha}$ is $\epsilon$-fractional-robust.
\end{theorem}
While the technique used in the following proof upper bounds the success-probability of an attack on an arbitrary block, we stress that its result cannot be interpreted as a security-guarantee for an arbitrary block. Transactions in this scheme are explicitly conditioned on the attack's state before the transmission of the transaction to the merchant. Nonetheless, this theorem shows that the merchant can guarantee any $\epsilon$-fractional robustness, by waiting long enough.
\begin{proof}
Let $b$ be an arbitrary block of the attacker. Let $T\gg0$ and denote by $N_T$ the total number of blocks created in the system up to time $T$. Finally, denote by $I_b$ the indicator random variable of the event where block $b$ participates in a successful attack. Under the generalized Vector76 attack, the attacker never publishes its secret chain. Therefore, eventually, all of the transactions in the honest network's chain will be accepted, as it grows indefinitely. Additionally, whenever an attacker block participates in a successful attack, by definition, the policy must have accepted a transaction in that block. Therefore, assuming a roughly constant number of transactions per block, and denoting the entire set of attacker blocks by $att$, we obtain:
\begin{align*}
& \lim\limits_{t\rightarrow\infty}\frac{\sum_{u\in \T^{t}}\Pr\left(\mathcal{E}_{attacked}(u)\right)}{height(t)}= \\ & \lim\limits_{t\rightarrow\infty}\frac{\frac1t\cdot\sum_{u\in \T^{t}}\Pr\left(\mathcal{E}_{attacked}(u)\right)}{\frac1t\cdot height(t)}= \\ & \lim\limits_{t\rightarrow\infty}\frac{\Pr_b\left(\mathcal{E}_{attacked}(b)\right)}{\frac1t\cdot height(t)},
\end{align*}
where the probability here is also over the choice of $b$ in $T^{\infty}$. This identity holds due to the Law of Large Numbers. The last expression is upper bounded by $\frac{\Pr_b\left(\mathcal{E}_{attacked}(b)\right)}{1-\alpha}$, as the longest chain grows at least at the rate of the honest network's chain.
We now provide an upper bound on $\Pr_b\left(\mathcal{E}_{attacked}(b)\right)$.

Fix the attacker's policy $\pi_A$.
The attacker can utilize $b$ to carry out a generalized Vector76 attack if and only if, at some point in time, the following conditions are met: $\sa>\sh$ and the number of blocks in the attacker's chain above $b$ is at least $k$. Indeed, if the first condition is not met then the merchant will count zero confirmations for its transaction in the honest network's chain, and will not accept. Likewise, if the second one is not met, the merchant will not see $k$ confirmations in the attacker's chain, and will not accept. Assume that block $b$ has $k$ confirmations (including itself). The probability that, in a period of time in which the attacker created $k$ blocks, the honest network created $n$ blocks is given by $\binom{n+k-1}{n}\cdot\alpha^k\cdot(1-\alpha)^n$.

As proven in Section~\ref{sec::premining}, the lead that the attacker gained over the network's chain, prior to mining $b$, can be upper bounded by the distribution vector $(p_l)_{l=0}^{\infty}$, with $p_l=\frac{1-2\cdot\alpha}{1-\alpha}\cdot\ratio^l$. In consequence, given $n$, the probability that the attacker will ever succeed in bypassing the honest network's chain is 1, if $l+k\ge n$, and $\ratio^{n-(l+k)}$, if $n>l+k$ (here we used the worst-case assumption that $\gamma=1$). Therefore, the probability that $b$ participates in a successful attack is upper bounded by $g(k,\alpha)$, which is the sum of $p_l\cdot \binom{n+k-1}{n}\cdot\alpha^k\cdot(1-\alpha)^n\cdot\ratio^{(n-(l+k))^+}$ over all $l$ and $n$.

Therefore, the expression~(\ref{eq::vec76defense}) upper bounds the probability that an arbitrary block of the attacker will have $k$ confirmations and at the same time will be part of a chain longer than the honest network's.
All in all, we obtain that the fraction of attacker blocks (out of the total number of accepted blocks) can be upper bounded by $\frac{g(k,\alpha)}{1-\alpha}$. In particular, $\sigma^{spv}_{\alpha}:=g(n,\alpha,\epsilon)$ is $\epsilon$-fractional-robust.
\end{proof}

\section{The profit of an attacker}\label{sec::profit_of_attacker}
Arguably, one might hope that carrying out double spending attacks would be costly to the attacker, due to the loss in potential profit the attacker could gain from participating honestly in the mining effort. This, presumably, will disincentivize attackers from committing to long attack-strategies which waste their resources. 
Alas, the work of Sapirshtein et. al.~\cite{OPT_SELFISH_MINING} observes that this is not the case in general. An attacker with sufficient hashrate or a significant $\gamma$ can actually combine profitable selfish mining with double spending attacks. Their idea is simple: Every profitable selfish mining scheme arrives with positive probability at states of the form $(\sa,\sh)$, $\sa>\sh+k$, from which the attacker can plan a definitely successful double spending attack (against the policy $\sigma\equiv k$). In this section we aim at quantifying the potential profit for the attacker under optimal combinations of these attacks.

To this end, we adapt the MDP used above to a setup in which the attacker maximizes its returns from both the block reward and fees, and from possible double spending attacks it manages to perform. We make the simplifying assumption that the reward is fixed between blocks, and thus we reward the attacker with one unit for every block of his which is part of the longest chain. We additionally reward it with $\mathcal{R}$ units of reward for every successful double spend. 
Following~\cite{ES}, we normalize the rewards by the length of the main chain.  
Figure~\ref{fig::cost3} depicts the profit of an attacker that carries out such an attack. The dashed line corresponds to honest mining without double spending attacks (where a miner with $\alpha$ of the hash rate gains an $\alpha$-fraction of the rewards), and the other curves correspond to the profit computed by the MDPs for different values of $\gamma$. Consistent with the results of~\cite{ES,OPT_SELFISH_MINING}, at $\gamma=1$ the attacker is always profitable (at any $\alpha$), and occasionally attacks with double spending attacks. Here, the gains from double spending are assumed to be $\mathcal{R}=2$ block rewards. 

\begin{figure}[ht]
	\centering
	\includegraphics[scale=0.8]{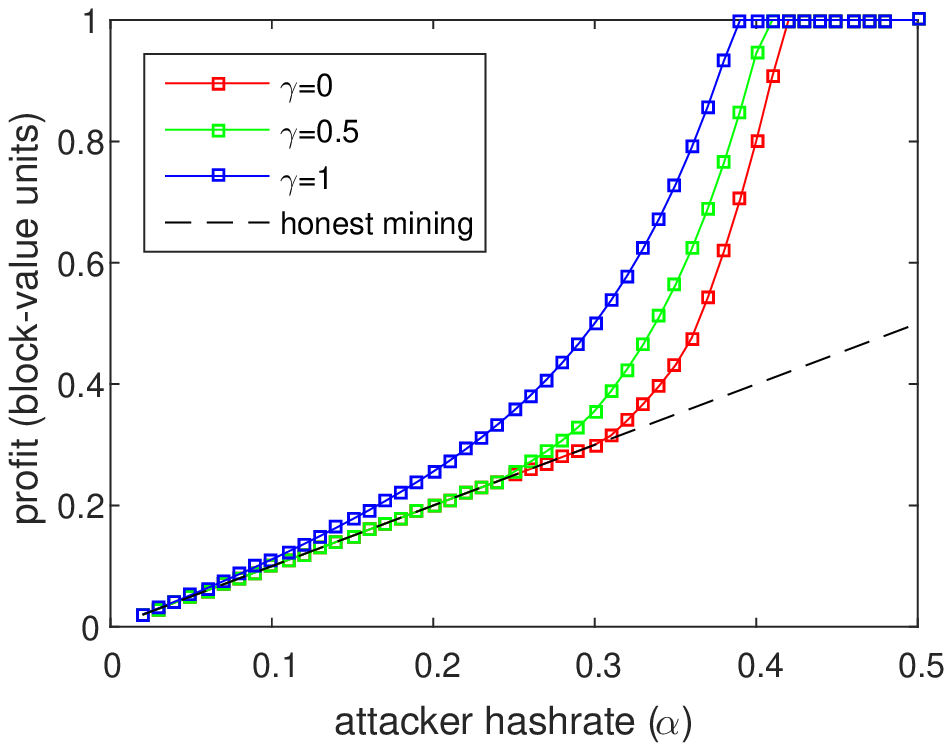}
	\caption{The profit of an attacker from carrying optimal combinations of selfish mining and double spend attacks on the $\sigma\equiv 6$ confirmations policy, assuming a successful double spent transaction is worth three times the value of an ordinary block.} 
	\label{fig::cost3}
\end{figure}

Figure~\ref{fig::cost05} depicts the gains for different values of double spend. The different plots correspond to different reward values from successful double spending. It is interesting to see the expected result: given that rewards from double spending increase, smaller miners can choose to deviate from honest behavior and gain (honest mining is again represented by the dashed line).

\begin{figure}[ht]
    \centering
    \includegraphics[scale=0.8]{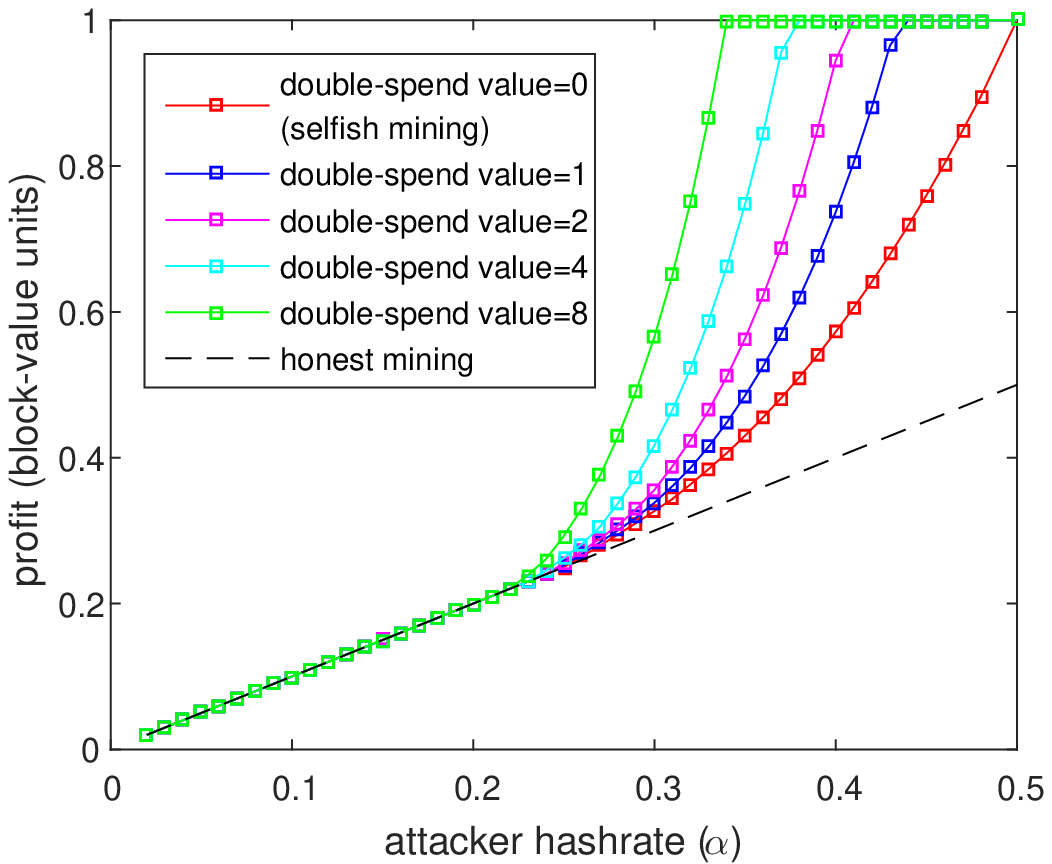}
  \caption{The profit of an attacker from carrying optimal combinations of selfish mining and double spend attacks on the $\sigma\equiv6$ confirmations policy, assuming that it is able to match a fraction of $\gamma=0.5$ of the honest nodes.} 
  \label{fig::cost05}
\end{figure}

\section{Conclusions}
We presented a variety of different interpretations of the security of a single transaction in the Bitcoin system, and matching advice regarding the number of confirmations merchants should await in order to properly secure their transactions. Our suggested prescription can be summarized, in short, as follows:
\begin{itemize}
\item Transactions whose timings can be assumed to be non-adversarial can be protected via $\epsilon$-arbitrary-robust policies, such as the $\sigma^{arb}$ family presented in Section~\ref{sec::premining}.
\item Merchants engaging in medium-valued transactions at a regular rate ought defend against a reversal of non-negligible fractions of the transactions they have authorized, in the long term (the $\epsilon$-fractional-robustness security model, coupled with the $\sigma^{frac}$ policy).
\item Recipients of large transactions are advised to commit to policies such as $\sigma^{total}$ (Section~\ref{sec::logarithmic}) which waits a time logarithmic in the chain's length, thereby guaranteeing themselves $\epsilon$-total-robustness, i.e., security from even a single reversal of any of their payments.
\item Light clients are advised to use a policy specifically protecting against the generalized Vector76 attack ($\sigma^{spv}$, Section~\ref{sec::vec76}).
\end{itemize}

Indeed, regarding the latter point, we demonstrated a clear case in which light nodes are less secure than full nodes solely due to the fact that they do not relay blocks further. This observation suggests several mitigation techniques, including sending requests for recent blocks and relaying them to the network, which would in fact imply that hybrids between light nodes and full nodes can be more secure. 

We have further shown that it is difficult to argue that attackers lose revenue from mining if they are trying to attack. Instead, Bitcoin can be considered to give guarantees lower bounding the losses of merchants.

Many research directions remain. The models here should be further adapted to settings with more significant delay. Several hybrid guarantees can also be explored. One such example is a fractional guarantee for a merchant that receives transactions occasionally (but not in every block), or attackers that are more limited in selecting the timing of their attacks (e.g., if a store is only open during the daytime). Finally, it would be interesting to evaluate the guarantees of variants of the protocol such as Bitcoin-NG~\cite{BitcoinNG} and Ethereum~\cite{ETHEREUM} against similar attacks, as they employ slightly different rules to manage the blockchain.

\bibliographystyle{abbrv}
\bibliography{security_model}

\appendix

\section{MDP description}\label{appendix::model}

In this section we describe briefly the computation method of the attack policy that maximizes the fraction of attacked blocks, against a defender policy of the form $\sigma_{\alpha}\equiv k$, for some constant $k$ (that may depend on $\alpha$). A block $b$ of the honest network is successfully attacked if the published chain above it is of length $k$ or more, including $b$ in the count, and the attacker then overrides it by publishing a longer chain (or matching it). A given state of the form $(\sa,\sh)$ represents the lengths of the attacker's and the network's chain, respectively, counted above the latest fork (i.e., the latest block adopted by both parties). Thus, if $\sa>\sh\ge k$, then the attacker can attack the $\sh+1-k$ blocks at the bottom of the honest chain (above the fork)---these blocks have $k$ confirmations, hence were accepted by the policy $\sigma$. The remaining $k-1$ blocks are indeed overridden but not attacked, since the policy didn't accept them yet, hence the transactions in them were not considered safe yet by their recipients. This is the main difference from selfish mining, where the attacker is rewarded for these blocks as well.
When the attacker abandons the attack and adopts, all blocks in the chain it adopted will be accepted and never attacked, since future attack blocks contain them in their history.

Accordingly, we grant the honest network a reward of $\sh$ whenever the attacker adopts. When the attacker adopts, we reward the attacker $\sh-(k-1)$ (this is the number of blocks it successfully attacked), and reward the honest network $k$ blocks (this complements the attacker's reward to the chain's new length $\sh+1$, and is needed for appropriate normalization). Further complexity arises due to the possibility of the attacker matching the honest chain's length (rather than succeeding it by 1). Whether this \ma action is feasible, for a given state $(\sa,\sh)$, is encoded in a third field called $fork$ with possible values: $irrelevant,relevant,active$.
For further details, and for a description of an algorithm that uses this MDP to maximize the fractional non-linear objective -- refer to~\cite{OPT_SELFISH_MINING}.

\begin{table}[h]
\caption{A description of the transition and reward matrices of the MDP. The third column contains the probability of transiting from the state specified in the left-most column, under the action specified therein, to the state on the second one. The corresponding two-dimensional reward (that of the attacker and that of the honest nodes) is specified on the right-most column.\label{table::PandR}}
\begin{tabular}{|K{37mm}|K{34mm}|K{23mm}|K{25mm}|}
\hline
\textbf{State $\times$ Action} & \textbf{State} & \textbf{Probability} & \textbf{Reward}\\
\hline
\multirow{2}{*}{$(\sa,\sh,\cdot),adopt$} & $(1,0,irrelevant)$ & $\alpha$ & \multirow{2}{*}{$(0,\sh)$}\\
\ & $(0,1,irrelevant)$ & $1-\alpha$ & \\
\hline
\multirow{2}{*}{$(\sa,\sh,\cdot),override^\dagger$} & $(\sa-\sh,0,irrelevant)$ & $\alpha$ & \multirow{2}{*}{$(\sh-k+1,k)^\mathsection$}\\
\ & $(\sa-\sh-1,1,relevant)$ & $1-\alpha$ & \\
\hline
\multirow{2}{*}{\makecell{$(\sa,\sh,irrelevant),wait$\ \\ $(\sa,\sh,relevant),wait$}} & $(\sa+1,\sh,irrelevant)$ & $\alpha$ & (0,0)\\
\ & $(\sa,\sh+1,relevant)$ & $1-\alpha$ & (0,0)\\
\hline
\multirow{3}{*}{ \makecell{$(\sa,\sh,active),wait$\ \\ $(\sa,\sh,relevant),match^\ddagger$}} & $(\sa+1,\sh,active)$ & $\alpha$ & (0,0)\\
\ & $(\sa-\sh,1,relevant)$  & $\gamma\cdot(1-\alpha)$ & $(\sh-k+1,k-1)^\mathparagraph$\\
\ & $(\sa,\sh+1,relevant)$ & $(1-\gamma)\cdot(1-\alpha)$ & (0,0)\\
\hline
\end{tabular}
\noindent\begin{scriptsize}$^\dagger$feasible only when $\sa>\sh$\end{scriptsize} \\ \noindent\begin{scriptsize}$^\ddagger$feasible only when $\sa\ge\sh$\end{scriptsize} \\ 
\noindent\begin{scriptsize}$^\mathsection$if $\sh<k-1$, then the reward is $(0,\sh+1)$, as no block was actually attacked.\end{scriptsize}\\ \noindent\begin{scriptsize}$^\mathparagraph$if $\sh<k-1$, then the reward is $(0,\sh)$, , as no block was actually attacked.\end{scriptsize}
\end{table}

\end{document}